\documentclass[11pt]{article}

\usepackage{mathtools}
\usepackage{amssymb}
\usepackage{amsthm}

\newcommand{\C}{\mathbb C}

\newcommand{\F}{\mathbb F}

\newcommand{\Q}{\mathbb Q}
\newcommand{\Z}{\mathbb Z}

\newcommand{\fp}{\mathfrak p}
\newcommand{\fP}{\mathfrak P}
\newcommand{\co}{\mathcal{O}}

\newtheorem{rmk}{Remark}

\newtheorem{lem}{Lemma}

\newtheorem{thm}{Theorem}

\begin{document}
\title{Nonexistence of two classes of generalized bent functions}
\author{Jianing Li, Yingpu Deng\\\
Key Laboratory of Mathematics Mechanization,\\
NCMIS, Academy of Mathematics and Systems Science,\\
Chinese Academy of Sciences, Beijing 100190, P.R. China\\
(Email: lijianing19891026@163.com, dengyp@amss.ac.cn)}
\date{}
\maketitle
\begin{abstract}
We obtain  new nonexistence results of generalized bent functions from $\Z^n_q$ to $\Z_q$ (called type $[n,q]$) in the case that there exist cyclotomic integers in $ \Z[\zeta_{q}]$ with  absolute value $q^{\frac{n}{2}}$. This result generalize the previous two scattered nonexistence results $[n,q]=[1,2\times7]$ of Pei \cite{Pei} and $[3,2\times 23^e]$ of Jiang-Deng \cite{J-D} to a generalized class. In the last section, we remark that this method can apply to the GBF from $\Z^n_2$ to $\Z_m$. 
\end{abstract}



\section{Introduction}

Bent functions were first introduced by Rothaus \cite{Rothaus} in
1976, which are functions from $\mathbb{Z}_2^n$ to $\mathbb{Z}_2$ with some property, where $\mathbb{Z}_m=\mathbb{Z}/m\mathbb{Z}$ for a positive integer $m$.
Dillon \cite{Dillon} showed that bent functions are the characteristic functions
of elementary Hadamard difference sets. Bent functions have many applications to coding theory,
cryptography and sequence designs \cite{OSW}. In coding theory, bent
functions have the maximum distance to the first order binary Reed-Muller code. In a cryptsystem, functions with large nonlinearity values are usually
employed to resist linear crypto-analysis and correlation-attack, and bent functions are just the ones with maximum nonlinearity.

Bent functions have many generalizations. Kumar et al. \cite{Kumar} considered bent functions from $\mathbb{Z}_q^n$ to $\mathbb{Z}_q$ in 1985, where $q$ and $n$ are positive integers. Later, bent functions were generalized
to arbitrary finite abelian groups \cite{lsy,solo}, even to arbitrary finite groups \cite{poin,xu}.

A natural question is when bent functions do exist. Rothaus \cite{Rothaus} proved that bent functions from $\mathbb{Z}_2^n$ to $\mathbb{Z}_2$ exist if and only if $n$ is even. However, this problem is far from being solved
for other type of generalized bent functions. For generalized bent functions (GBF for short) from $\mathbb{Z}_q^n$ to $\mathbb{Z}_q$ defined in \cite{Kumar}, Kumar et al. constructed them except the case that $n$ is odd and $q\equiv2\pmod 4$. There are many nonexistence results for GBF defined in \cite{Kumar}, for example, see \cite{Pei,Ikeda,Feng,Feng2,Feng3,J-D} and the references in \cite{J-D}. Especially, Feng and co-authors built connections between nonexistence of GBF and class numbers of imaginary quadratic fields in \cite{Feng,Feng2,Feng3}. In fact, they proved stronger results, i.e., there are no algebraic integers with prescribed absolute values in some cyclotomic field. However, GBFs' existence request that there are algebraic integers with prescribed absolute values in some cyclotomic field and some specific conditions (so-called bent conditions) are also satisfied. In \cite{J-D},
Jiang and Deng showed that there are no GBFs from $\mathbb{Z}_q^n$ to $\mathbb{Z}_q$ with $n=3$ and $q=2\cdot 23^e$ for $e\geq1$. Notice that there are algebraic integers with prescribed absolute values 
$(2\cdot 23^e)^{\frac{3}{2}}$ in the cyclotomic field $\mathbb{Q}(\zeta_{23^e})$, but Jiang and Deng showed that the bent conditions are not satisfied.

Motivated by the results in \cite{Ikeda} and \cite{J-D}, we generalize the method used there and we obtain further nonexistence results for GBFs from $\mathbb{Z}_q^n$ to $\mathbb{Z}_q$. For the details, see the following section.

This paper is organized as follows. In Section 2, we list some previous work and state our main result. We prove the main result in Section 3. In section 4, we apply our method to GBFs from $\mathbb{Z}_2^n$ to $\mathbb{Z}_m$ and obtain similar results. Finally, a short
conclusion is given.

\section{Previous Work and Our Main Result}

Let $q\geq 2$ be an integer, $\Z_q=\Z/{q\Z}$, $\zeta_q=\text{exp}(\frac{2\pi i}{q})$. A function $f$ from $\Z^n_q$ to $\Z_q$ 
is called a Generalized Bent Function(GBF) with type $[n,q]$ if $F(\lambda)\overline{F(\lambda)}=q^n$ for every $\lambda\in \Z^n_q$ where
$$ F (\lambda)=\sum_{x\in\Z^n_q}\zeta^{f(x)}_q\cdot \zeta^{-x\cdot \lambda}_q $$ is the Fourier transform of the function $\zeta^{f(x)}_q$, $x\cdot \lambda$ is the standard dot product, and $\overline{F(\lambda)}$ is the complex conjugate of $F(\lambda)$.

Note that, if there is no element in $\Z[\zeta_q]$ with absolute value $q^{\frac{n}{2}}$, then there is no GBF with type $[n,q]$.
 Feng, and Feng-Liu \cite{Feng}\cite{Feng2}\cite{Feng3} get nonexistence results by showing that there is no cyclotomic integers with prescribe absolute values.  For a survey of their results, see \cite{J-D}. Here we  list a special case which is  proved by Feng: 
 \\

 (1) Let $p\equiv7 \pmod8$ be a prime, $f$ the order of $2$ modulo $p^e$, and $m$ the smallest odd positive integer such that $x^2+py^2=2^{m+2}$ has $\Z$-solutions, then there is no GBF with type $[n<\frac{m}{s}, 2 p^e]$, where  $n$ is  odd and $s=\frac{p^e-p^{e-1}}{2f}$ .

 However, in this paper we focus on the nonexistence  of GBF  in the case that \textbf{there exist cyclotomic integers with prescribe absolute value}. Toward this direction, there are only three results:
\\

  (2) Pei(1993)\cite{Pei} proved that there is no GBF with type $[n=1,q=2\times 7]$. Notice that  the absolute value of $(\frac{1+\sqrt{-7}}{2}) \sqrt{-7}\in\Z[\zeta_7]$ is $14^{\frac{1}{2}}$.
 \\

 (3) Ikeda(1999)\cite{Ikeda}  proved that there is no GBF of type $[n=1,q=2\times p^e]$, where $p\equiv 7 \pmod8$ is a prime. This result contains both cases there are or not cyclotomic integers with prescribe absolute values.  For example if $p=7$, the absolute value of $(\frac{1+\sqrt{-7}}{2}) (\sqrt{-7})^e\in \Z[\zeta_{7^e}]$ is $(2\times7^e)^{\frac{1}{2}}$. If $p=23$, there is no cyclotomic integers in $\Z_[\zeta_{23}]$ with absolute value $2\times 23$. Of course, this theorem include Pei's result.
 \\
 
 (4) Jiang-Deng(2015)\cite{J-D} proved that there is no GBF with type $[n=3,q=2\times 23^e]$. Notice that  the absolute value of $(\frac{3+\sqrt{-23}}{2}) (\sqrt{-23})^e\in \Z [\zeta_{23^e}]$ is $(2\times23^e)^\frac{3}{2}$.
\\

In this article we extend the  results (2), (4) to a general situation by generalizing the method developed in \cite{Ikeda} and \cite{J-D}.  

We need a definition to state our result. Let $p\equiv7\pmod 8$ be a prime number, then $2$ splits in $\mathbb{Q}(\sqrt{-p})$. Let $ \fp$ be a  prime ideal of $\Q(\sqrt{-p})$ above $2$. We define $t_p$  to be  the minimal positive integer such that $\fp^{t_p}$ is pricinpal in this article. By Gauss's genus theory, $t_p$ is odd. For more about $t_p$, see Remark~\ref{basic fact of t_p} below. In this article, our main result is the  following:

\begin{thm} \label{main}
If $p\equiv7 \pmod 8$ is a prime and  the order of $2$ modulo $p$ is $\frac{p-1}{2}$, then there doesn't exist GBF with type $[n=t_p,q=2p]$. If $p$ further satisfies $2^{p-1}\not\equiv 1 \pmod {p^2}$, then there doesn't exist GBF with type $[n=t_p,q=2p^e]$ for every $e\geq1$.		
\end{thm}

\begin{rmk}
	The condition $2^{p-1}\not\equiv 1\pmod{p^2}$ is for forcing the inert degree of $2$ in $\Q(\zeta_{p^e})$ is $\frac{\phi(p^e)}{2}$, where $\phi$ is Euler function. For $p<6\times10^9$, there is only one $p=3511 \equiv 7\pmod 8$ such that $2^{p-1}\equiv 1\pmod{p^2}$.
	
\end{rmk}

\begin{rmk}\label{basic fact of t_p}
	\textbf{Some basic facts about the number $t_p$}
	
	Of course, $t_p$ is just the order of $\fp$ in the class group of $\Q(\sqrt{-p})$. The definiton of $t_p$ does not depend on the chioce of $\fp$.  Another prime ideal above $2$ is the inverse of $\fp$ in the class group, so they have the same order. By Gauss's genus theory or class field theory \cite{Wash} Theorem 10.4(b), the class number of $\Q(\sqrt{-p})$ is odd, hence $t_p$ is also odd.  In \cite{Feng} Remark 2, Feng gave an  estimate that $t_p>\frac{\log{p}}{\log 2}-2$. In particular, $t_p\geq3$, if $p\geq23$. He also showed that $t_p$ is the smallest odd positive integer such that $x^2 + py^2 = 2^{t_p +2}$ has solutions with $(x,y)\in \Z^2$. So, for computing $t_p$, we can first caluate the class number, then for each divisor $>\frac{\log{p}}{\log 2}-2 $ check out this equation.
\end{rmk}	
	
We give a small table of the primes less than $200$ and  satisfy all conditions in Theorem~\ref{main}. From the table below, it can be seen that our result includes the results of Pei and Jiang-Deng.

	\begin{tabular}{l*{9}{c} r}
		$p$   &7 &23 &47 &71 &79 &103 &167 &191 &199 \\
		\hline 
		$t_p$  &1 &3 &5 &7 &5 &5 &11 &13 &9 \\

	\end{tabular}

 By the estimate in Remark~\ref{basic fact of t_p} we know  that $t_p=1$ if and only if $p=7$.  So our result is different from Ikeda's.

Now we explain  the condition $n=t_p$ briefly.
Fix a prime $p$ satisfying all  conditions in the Theorem~\ref{main}. Let us focus on the nonexistence of GBF with type $[n,q=2p^e]$ and  let $n$ always be odd.

 If $n<t_p$, Feng's above result (1) shows that there does not exist an element in $\Z[\zeta_{p^e}]$ with absolute value $q^{n/2}$, so there is no GBF with type $[n,q=2p^e]$ with $n<t_p,e\geq1$. Therefore, Feng's result and our result
have shown that there is no GBF with type $[n,2p^e]$ for every odd $n\leq t_p, e\geq1$.

If $n\geq t_p$, the case is different from the above. Because there exist cyclotomic integers with absolute value $q^{n/2}$, see Lemma~\ref{n=t}. Our result treats the case $n=t_p$.  However, the case $n>t_p$ is unknown even there is no single example be checked.

\section{Proof of Theorem~\ref{main}}

 For the basic facts of algebraic number theory used in this paper, one can found them in any standard textbook of algebraic number theory, for example, \cite{Marcus,Flo}.

 Let $p\equiv7 \pmod8$ be a prime  such that the order of $2$ modulo $p$ is $\frac{p-1}{2}$, $\zeta=\zeta_{p^e}=\text{exp}(\frac{2\pi i}{p^e})\in \C$,  $L=\Q(\sqrt{-p})$ and $K=\Q(\zeta)$, then $L\subset K$.  We have the  prime ideal factorization of $p$ and $2$ in $L$ and $K$ as follows. 
 
 For $p$, we have $(p)\co_K=(1-\zeta)^{\phi(p^e)}\co_K$, and $\sqrt{-p}\co_K=(1-\zeta)^\frac{\phi(p^e)}{2}\co_K$.

 For $2$, 
 the conditions in the Theorem~\ref{main} implies that the order of $2$ modulo $p^e$ is $\frac{\phi(p^e)}{2}$ for each $e$, by the cyclotomic reciprocity law  $2\co_L=\fp\co_L \overline{\fp}\co_L$,
 where $\bar{\fp}$ is the complex conjuate of $\fp$, and $\fp=\fP\co_K,\overline{\fp}=\overline{\fP}\co_K$ are inert in $\Q(\zeta)$, and $L$ is the decomposition field of $2$ in $K$, also the fixed field by Frobenieus of $2$.

 Let $\gamma\in \co_L$ such that $\fp ^t=(\gamma)$, then $\overline{\fp} ^t=(\overline{\gamma})$, where $t=t_p$ for simplicity, then $$\gamma \overline{\gamma}=2^t \text{  as elements.}$$ Since the units of $\co_L$ are $\pm1$, the elements in $\co_L$ with absolute value $2^\frac{t}{2}$ are $\pm\gamma$ and $\pm\bar{\gamma}$.

Now assume $f$ is a GBF with type $[t=t_p,q=2p^e]$, then $F(\lambda)\in \co_K$ with absolute value $q^{\frac{t}{2}}$. The following two lemmas tell us the form $F(\lambda)$ must be. The first lemma is a slight refined version of \cite{Feng} Lemma 2.2.
\begin{lem}
	If $\alpha \overline{\alpha}=2^n$ for some $\alpha\in \co_K$, then there exist a unique $i \pmod {p^e}$ such that $\alpha \zeta^i \in \co_L$. 
	
\end{lem}

\begin{proof}
	Let $\sigma_2\in \text{\text{\text{Gal}}}(K/\Q)$ be the Frobenius of $2$, so $\sigma_2$ fix $\fP$ and $\overline{\fP}$. Since the prime ideal factors of $\alpha$ is $\fP$ and $\overline{\fP}$, so we have ideal equality:
	$$(\sigma_2 (\alpha))=(\alpha),$$
	
	then there is a unit $u$ of $\co_K$ such that $\sigma_2(\alpha)=u\alpha$, then by $\text{\text{\text{Gal}}}(K/Q)$ is Abelian, we get 
	$$\sigma(\alpha)\overline{\sigma_2(\alpha)}=u\alpha \overline{u}\overline{\alpha}=2^t.$$
	
	Then $u\bar{u}=1$, and by the fact that $\text{\text{Gal}}(K/\Q)$ is abel again, we get $\sigma(u) \overline{\sigma(u)}=1$ for every $\sigma\in \text{\text{Gal}}(K/\Q).$ Therefore, $u$ is a root of unity, hence there is a unique $i \pmod{p^e}$ such that $u=\pm\zeta^i$. Let $\beta=\alpha \zeta^{-i}$, then 
	$$\sigma_2 (\beta)=\alpha(\pm\zeta^{-i}),$$ so $\sigma_2(\beta^2)=\beta^2$, since $L$ is the fixed field by $\sigma_2$, this implies $\beta^2\in \co_L$, but $[K:L]=\frac{p-1}{2}$ is odd, so $\beta=\alpha\zeta^{-i}\in \co_L$.
	
\end{proof}

\begin{lem} \label{n=t}
	$F(\lambda)$ must be one of the following elements $$(\sqrt{-p})^{et}\cdot \gamma \cdot (\pm\zeta^i)$$ or
	$$(\sqrt{-p})^{et}\cdot \overline{\gamma} \cdot (\pm\zeta^i),$$ 
	 where $i\in\{0,1,\cdots,p^e-1\}$. In particular, $F(\lambda)\notin (2)\co_K=\fP\overline{\fP}\co_K$.
\end{lem}

\begin{proof}
	It's easily seen that every element above has absolute value $q^{\frac{t}{2}}$.
	
	By consider their prime ideal factorizations we have 
	$$(F(\lambda)\overline{F(\lambda)})=(q^t)=(2^tp^{et})\co_K=\fP^t{\overline{\fP}}^t(1-\zeta)^{\phi(p^e)et}.$$
	
	Note that $\overline{(1-\zeta)}=(1-\zeta)$, we know the exact power of $(1-\zeta)$ in $(F(\lambda))$ must be $\frac{\phi(p^e)et}{2}$, so  $(1-\zeta)^{\frac{\phi(p^e)et}{2}}=(\sqrt{-p})^{et}  ||(F(\lambda))$ as ideals of $\co_K$. Hence $\alpha: =\frac{F(\lambda)}{\sqrt{-p}}^{et} \in \co_K$ and $\alpha\overline{\alpha}=2^t$.  Then by the above lemma, we konw there is a unique $i \pmod {p^e}$ such that $\beta:=\alpha {\zeta^{-i}} \in \co_L$. Also, $\beta \overline{\beta}=2^t$, and the only elements in $\co_L$ with absolute value $2^\frac{t}{2}$ are $\pm\gamma$ and $\pm\bar{\gamma}$, so $\beta=\pm \gamma$ or $\pm\overline{\gamma}$. Hence$$F(\lambda)=(\sqrt{-p})^{et}\cdot \alpha=(\sqrt{-p})^{et}\cdot \beta \zeta^{i}=(\sqrt{-p})^{et}\cdot \gamma \cdot (\pm\zeta^i)$$ or $$(\sqrt{-p})^{et}\cdot \overline{\gamma} \cdot (\pm\zeta^i)$$ for some $i$.
\end{proof}

\begin{lem} \label{f(x)=f(x+v)}
	If $v\in \Z^t _q$ is an element of order 2, then for every $\lambda \in \Z ^t _q$,  $$F(\lambda)=\pm F(\lambda+v).$$
\end{lem}

\begin{proof}
	(This proof is essencially 
	same as in \cite{Ikeda} Lemma 3 and \cite{J-D} Lemma 8).
	
	As $v$ is of order $2$, for $x\in \Z^t_q$, we have $\zeta^{x\cdot v}_q=\pm1$. So $$F(\lambda)+F(\lambda+v)=\sum_{x\in \Z^t_q}{\zeta^{f(x)-x\cdot \lambda}(1+\zeta^{-x\cdot v})}\in (2)=\fP\cap\overline{\fP},$$ or said differently, $$(2) | (F(\lambda)+F(\lambda+v))$$ as ideals in $\co_K$.
	
	Then $F(\lambda)\in \fP \iff F(\lambda+v)\in \fP$. By the above Lemma, we have $F(\lambda+v)=F(\lambda)\cdot(\pm\zeta^i)$ for some $0\leq i< p^e$. Then $$F(\lambda)+F(\lambda+v)=F(\lambda)(1\pm\zeta^i).$$ Note that  if $i\neq0$, $(2)$ and $(1\pm \zeta^i)$ are coprime ideals. So $(2)|(F(\lambda))$. It is a contradition to the form of $F(\lambda)$.  Therefore, $i=0$.

\end{proof}

Let $G$ be the Sylow-$2$ subgroup fo $\Z ^t_q$, then $G\cong \F^t_2$ ,write $G=\{0,v_1,\cdots,v_{2^t-1}\}$. For every $v_i\in G$,  we define

$$N_i=N_{v_i}=\{x\in \Z^t_q| F(x)=F(x+v_i)\},$$
$$M_i=M_{v_i}=\{x\in \Z^t_q| F(x)=-F(x+v_i)\}.$$

So $\Z^t_q=N_i\sqcup M_i$, the symble $\sqcup$ means no intersection union.
Let $n_i=|N_i|,m_i=|M_i|$ be the cardinality of $N_i,M_i$, so $n_i+m_i=q^t$.
\begin{lem} \label{n=m}
	$n_i=m_i=\frac{q^t}{2}$.
\end{lem}
\begin{proof}
	As $F$ is a Fourier transform of $\zeta^{f(x)}$, which is a function from $\Z^t_q$ to $\mathbf{S}^1=\{z\in \C: |z|=1 \}$, we have $$\sum_{x\in \Z^t_q}{F(x)\overline{F(x+v_i)}}=0.$$Then by the definition of GBF we have  $F(x)\overline{F(x)}=q^t$ for each $x$, so we have
	
	$$0=\sum_{x\in N_i}{F(x)\overline{F(x+v_i)}}+\sum_{x\in M_i}{F(x)\overline{F(x+v_i)}}=q^t(n_i-m_i). $$ Hence, $n_i=m_i$ for each $i$.
\end{proof}

Here we review the  proof of  Ikeda \cite{Ikeda} for the case $t=1$, since it is the basic idea of this method. First note that if $x\in N_v$, $x+v\in N_v$, where $v$ is the unique element of order $2$ in $\Z_q$, hence $2|n_v$. But $n_v=\frac{q}{2}$ is odd. This is a contradiction.

It can be seen that Ikeda \cite{Ikeda} used one $2$-order element to prove the case $t=1$.  In Jiang-Deng \cite{J-D}, they use three $2$-order elements to treat the case $t=3$. For the remaining of this article, the goal is to generalize this method systematically to treat the general case.

In the remain of the proof, we assume  $t\geq 3$ so that $|G|\geq8$. Now we are going to define $2^{2^t -1}$ subsets of $\Z^t_q$ by using all $2$-order elements as follows:
$$X_1\cap X_2\cap \cdots \cap X_{2^t-1},$$ each $X_i=N_i \text{ or }M_i$.
Obviously, $\Z^t_q$ is a no intersection union of these subsets. Our main task is to compute the cardinality of each subsets.

\begin{lem}
	If $u,v,w\in G-\{0\}$ are pairly different and $u+v+w=0$, then we have $$N_u\cap N_v\cap M_w=N_u\cap M_v\cap N_w=M_u\cap N_v\cap N_w=M_u\cap M_v\cap M_w=\emptyset.$$
\end{lem}

\begin{proof}
	(We only need the fact $F(\lambda)\notin (2)\co_K$, so the proof below is essentially same as the case $p=23$ as in \cite{J-D} Lemma 11, though they choose three special $2$-order elements there).
	
	First, note that $$x\in N_u\cap N_v\cap M_w,$$
	$$\iff x+u \in N_u\cap M_v\cap N_w,$$
	$$\iff x+v \in M_u\cap N_v\cap N_w,$$
	$$\iff x+w \in M_u\cap M_v\cap M_w.$$
	
	So it's enough to prove $M_u\cap M_v\cap M_w=\emptyset$.

	Second, note that the map $$(    \text{    }\cdot u): \Z^t_q \longrightarrow \{0,\frac{q}{2}\}\subset\Z_q$$
	
	$$y\mapsto y\cdot u$$ is surjective. So $\zeta^{y\cdot u}=0$ if $y\cdot u=0$ and $\zeta^{y\cdot u}=-1$ if $y\cdot u\not=0$
	
	 For simplicity, we let $\sum_{y}=\sum_{y}{\zeta^{f(y)-x\cdot y}}$.
	 Now take an element $x\in M_u\cap M_v\cap M_w$, then
	$$F(x)=\sum_{y\in \Z^t_q}=\sum_{y\cdot u=0}+\sum_{y\cdot u\neq0}$$
	$$=\!=\!-F(x+u)=-\sum_{y\cdot u=0}+\sum_{y\cdot u\neq0}.$$
	
	So we get $$0=\sum_{y\cdot u=0}=\sum_{y\cdot u=0, y\cdot v=0	}+\sum_{y\cdot u=0,y\cdot v\neq0}.$$
	
	Similarily, we have $$0=\sum_{y\cdot v=0}=\sum_{y\cdot u=0, y\cdot v=0	}+\sum_{y\cdot u\not=0,y\cdot v=0}.$$
	
	Also since $$0=\sum_{y\cdot w=0},$$ we have that
	$$F(x)=\sum_{y\cdot w\not=0}=\sum_{y\cdot (u+v)\neq0}=\sum_{y\cdot u=0,y\cdot v\neq0}+\sum_{y\cdot u\not=0,y\cdot v=0}=-2\sum_{y\cdot u=0,y\cdot v=0}\in (2)\co_K.$$ This contradict to the Lemma ~\ref{n=t}.	 So $M_u\cap M_v\cap M_w=\emptyset$.
	
\end{proof}

 The following observation (Lemma~\ref{observation}) will make us  generalize  Ikeda's and Jiang-Deng's method systematically, since it tells us among the $2^{2^t-1}$ sets, there are at most $2^t$ nonempty sets. And  these $2^t$ sets are rather ``nice". 

\begin{lem} \label{observation}
	Let $N_i,M_i$ be as above, $X_1\cap X_2\cap \cdots \cap X_{2^t-1}\subset \Z^t_q$ with $X_i=N_i$ or $M_i$. If $\{v_i\in G|X_i=N_i\}\cup\{0\}$ is not a subgroup of $G$ with index $1$ or $2$, then it must be empty.
\end{lem}

\begin{proof}
	If  $A:=\{v_i\in G|X_i=N_i\}\cup\{0\}$ is not a subgroup, then there are $u,v\in A$ such that $u+v\notin A$. Then by the above Lemma  $$X_1\cap X_2\cap \cdots \cap X_{2^t-1}\subset N_u\cap N_v\cap M_{u+v}= \emptyset.$$ So $A$ is a subgroup and also a $\F_2 $ vector subspace of $G$. If the index of $A$ is larger than $2$, then its $\F_2$-dimension $\leq {t-2}$, then the dimension of its orthgonal complenent subspace  $\geq 2$, then at least there are three different element $u,v,w\notin A$. Then by the above Lemma
		$$X_1\cap X_2\cap \cdots \cap X_{2^t-1}\subset M_u\cap M_v\cap M_w= \emptyset.$$

\end{proof}

\begin{proof}[Proof of The Theorem~\ref{main}]

There are $2^t-1$ subgroups of $G$ with index $2$, because the $$ \{\text{subgroups of index 2}\}=$$$$\{ \F_2 \text{ vector subspace of G with dimsension } t-1\} \longleftrightarrow  \{1-\text{dim } \F_2 \text{vector subspace of } G\} $$ by taking the orthgonal complement subspace. Denote them by $H_1,\cdots,H_{2^t -1}$, we sometimes use the symble $H$ to denote some $H_i$.

Let $$N_G=\bigcap_{i=1}^{2^t-1} {N_i},$$
$$N_{H_i}=(\bigcap_{v\in H_i-\{0\}} N_{v} )\bigcap (\bigcap_{u\notin H_i} M_u)$$
Let $n_G=|N_G|$, $h_i=h_{H_i}=|N_{H_i}|$.

From $$\Z^t_q=(\bigsqcup_{i=1}^{2^t-1} N_{H_i})\bigsqcup N_G$$
  we get an equation
\begin{equation}
\sum_{i=1}^{2^n-1} h_i +n_G=q^t
\end{equation}

 On the other hand $$N_i=\bigsqcup_{H\ni v_i} N_H$$ for every $i$ satisfies ${1\leq i\leq 2^t-1}$, we get $2^t-1$ equations:

$$n_i=\frac{q^t}{2}=\sum_{H\ni v_i} h_H+n_G$$ for $i=1,2,\cdots,2^t-1.$

Note that each $H_i-\{0\}$ has $2^{t-1}-1$ elements. Summing  these $2^{t-1}-1$ equations we get
$$(2^t-1) \frac{q^t}{2}=(2^{t-1}-1)\sum_{i=1}^{2^t-1}{h_i}+(2^t-1)n_G$$

Combine with (1), we get $n_G=\frac{q^t}{2^t}=p^{te}$ is an odd number.

However, if $x\in N_G$, then we have $x+v\in N_G$ for all $v\in G$, so $2^t|n_G$. This contradiction shows that there doesn't exist GBF with type $[t=t_p,q=2p^e]$ for any $e\geq 1$, and this completes the proof of our Theorem~\ref{main}.
\end{proof}

\section{GBF From $\Z^n_2$ to $ \Z_m$}

In this section, we consider another kind of  generalized bent function (GBF), and use the same method of the above section we get some similar nonexistence result. 

A function $f$ from $\Z^n_2$ to $\Z_m$ is called a GBF with type $[n,m]$ if the equality holds  
$$|F(\lambda)|:=|\sum_{x\in \Z^n_2}\zeta^{f(x)}_m\ (-1)^{x\cdot \lambda}|=2^{\frac{n}{2}}$$ for every $\lambda\in\Z^n_2$, where $x\cdot \lambda$ is the standard dot product.

There are many constructions of this type GBF, for example see\cite{smgs}. For nonexistence results, a good reference is \cite{Liu-Feng},
where Liu-Feng-Feng  proved many nonexistence results of this type GBF by showing there is no cyclotomic integers with prescribe absolute values (in our case, there exist cyclotoimc integers with prescribe absolute values).

By the same method used in the Section 3, we proved the following theorem :

\begin{thm} \label{thm2}
If $p\equiv7 \pmod 8$ is a prime and  the order of $2$ modulo $p$ is $\frac{p-1}{2}$, then there doesn't exist GBF with type $[n=t_p,m=2p]$. If $p$ further satisfies $2^{p-1}\not\equiv 1 \pmod {p^2}$, then there doesn't exist GBF with type $[n=t_p,m=2p^e]$ for every $e\geq1$.	
	
\end{thm}

The proof of Theorem ~\ref{thm2} is as same as the proof of Theorem~\ref{main}, we only give a skectch of this proof:

\begin{proof}
	
	Let $f$ be a GBF from $\Z^n_2$ to $\Z_m$, then $F(\lambda)\overline{F(\lambda)}=2^{\frac{t}{2}}$, where $t=t_p$. The notations below is the same in the above section.

	By tracing the proof of Lemma 1,2 we get that 
	$F(\lambda)$ must be one of the following elements $$ \gamma \cdot (\pm\zeta^i)$$ or
$$ \overline{\gamma} \cdot (\pm\zeta^i)$$ 
$i=0,1,\cdots,p^e-1$. In particular, $F(\lambda)\notin (2)\co_K=\fP\overline{\fP}\co_K$. 

Once we get the form of $F(\lambda)$ must be, the remain proof is exactly the  same as before.

\end{proof}

\section{Conclusion and future work}
In this article we give some new results of nonexistence of GBF on the condition that there exist cyclotomic integers with prescribe absolute values. Of course, the nonexistence problem is far from solved. Our future work toward the case when the order of $2$ modulo $p$ is less than $\frac{p-1}{2}$. For example the case $p=31,t_{31}=3$, we don't konw whether there exists GBF from $\Z^3_q$ to $\Z_q$ with type $[3,q=2\times 31]$ at present.

\vspace{5mm}

\noindent \textbf{Acknowledgments}\quad  We thank Yupeng Jiang, Chang Lv and Jiangshuai Yang for helpful discussions. The work of this paper
was supported by the NNSF of China (Grant No. 11471314) and the National Center for Mathematics and Interdisciplinary Sciences, CAS.


\begin{thebibliography}{9}
\bibitem{Dillon}  Dillon J.F.: Elementary Hadamard difference sets. PhD dissertation, University of Maryland, (1974).

	\bibitem{Feng} Feng K.: Generalized bent functions and class group of imaginary quadratic fields. Sci. China A \textbf{44}, 562--570 (2001).

    \bibitem{Feng2} Feng K., Liu F.: New results on the nonexistence of generalized bent functions. 
                    IEEE Trans. Inform. Theory {\bf49}, 3066--3071(2003a).

    \bibitem{Feng3} Feng K., Liu F.: Non-existence of some generalized bent functions. 
                         Acta Math. Sin.(Engl. Ser.). {\bf19}, 39--50(2003b).

    \bibitem{Flo}  Fr\"ohlich A., Taylor M.J.: Algebraic Number Theory. Cambridge University Press, Cambridge (1991).
	
	\bibitem{Ikeda}Ikeda M.:  A remark on the non-existence of generalized bent functions. Lect.Notes Pure Appl. Math. \textbf{204}, 109--119(1999).

	\bibitem{J-D} Jiang Y., Deng Y.: New results on nonexistence of generalized bent functions. Des. Codes Cryptogr.  \textbf{75}, 375--385(2015).

    \bibitem{Kumar} Kumar P.V., Scholtz R.A., Welch L.R.: Generalized bent functions and their properties.
                  J. Comb. Theory Ser. A {\bf40}, 90--107(1985).

     \bibitem{Liu-Feng}  Liu H., Feng K., Feng R.: Nonexistence of Generalized Bent Functions From $\mathbb{Z}_{2}^{n}$ to $\mathbb{Z}_{m}$.  preprint, 2015.
     
    \bibitem{lsy} Logachev O. A., Salnikov A.A, Yashchenko V.V.: Bent Functions Over a Finite
              Abelian Group. Discrete Math. Appl. \textbf{7}, 547--564(1997).
     
     \bibitem{Marcus}  Marcus D. A.: Number Fields. Springer-Verlag, Berlin(1997).
     
     \bibitem{OSW}  Olsen J. D.,  Scholtz R. A,   Welch L. R.: Bent-function sequences.
             IEEE Trans. Inform. Theory {\bf28}, 858--864(1982).
	
	\bibitem{Pei} Pei D. On nonexistence of generalized bent functions. Lect. Notes Pure Appl. Math. \textbf{141}, 165--172 (1993).
	
    \bibitem{poin}  Poinsot L.: Bent functions on a finite nonabelian group. J. Discrete Math. Sci. Cryptogr.
      \textbf{9}, 349--364(2006).

	\bibitem{Rothaus}  Rothaus O. S.: On ``bent" functions.
           J. Comb. Theory. A {\bf20}, 300--305(1976).

    \bibitem{solo}  Solodovnikov V. I.: Bent Functions from a Finite Abelian Group to a Finite Abelian Group.
                Diskret. Mat. 14, 99--113(2002).

    \bibitem{smgs}  St\u{a}nic\u{a} Pantelimon.,  Martinsen Thor.,  Gangopadhyay Sugata.,  Singh Brajesh Kumar.: Bent and generalized bent Boolean functions. Des. Codes Cryptogr. \textbf{69}, 77--94(2013).	

	\bibitem{Wash} Washington L. C.: Introduction to Cyclotomic Fields. 2nd edition.
	    Graduate Texts in Math. 83, Springer-Verlag, New York, Heidelberg,
	    Berlin, (1997).

    \bibitem{xu} Xu B.: Bentless and nonlinearity of functions on finite groups. Des. Codes Cryptogr.
     (2014), DOI 10.1007/s10623-014-9968-y.
	

\end{thebibliography}
\end{document}